\documentclass[prodmode,acmtoms]{acmsmall} 

\usepackage{graphicx}
\usepackage{caption}
\usepackage{subfigure}
\usepackage{setspace}
\usepackage{tabulary}
\usepackage{fancybox}
\usepackage{enumitem}

\usepackage{courier}

\makeatletter
\newenvironment{CenteredBox}{%
\begin{Sbox}}{
\end{Sbox}\centerline{\parbox{\wd\@Sbox}{\TheSbox}}}
\makeatother

\usepackage{amsthm}

\usepackage{lineno}
\usepackage{xfrac}
\usepackage{amssymb}

\newtheorem{Def}{Definition}
\newtheorem{Prop}{Proposition}
\newtheorem{Algo}{Algorithm}
\newtheorem{Const}{Constraint}
\newtheorem{Strategy}{Strategy}
\newtheorem*{Rem*}{Remark}

\usepackage{alltt}

\usepackage{listings}
\usepackage[cmex10]{amsmath}
\usepackage{url}

\usepackage[utf8]{inputenc}

\DeclareMathOperator*{\argmin}{arg\,min}

\begin{document}

\markboth{F. Luporini et al.}{An algorithm for the optimization of finite element integration loops}


\title{An algorithm for the optimization of finite element integration loops}
\author{Fabio Luporini
\affil{Imperial College London}
David A. Ham
\affil{Imperial College London}
Paul H. J. Kelly
\affil{Imperial College London}}

\begin{abstract}
We present an algorithm for the optimization of a class of finite element integration loop nests. This algorithm, which exploits fundamental mathematical properties of finite element operators, is proven to achieve a locally optimal operation count. In specified circumstances the optimum achieved is global. Extensive numerical experiments demonstrate significant performance improvements over the state of the art in finite element code generation in almost all cases. This validates the effectiveness of the algorithm presented here, and illustrates its limitations. 
\end{abstract}





\acmformat{Fabio Luporini, David A. Ham, and Paul H. J. Kelly, 2016. An algorithm for the optimization of finite element integration loops.}


\begin{bottomstuff}

This work was supported by by the
Department of Computing at Imperial College London, the Engineering and
Physical Sciences Research Council [grant number EP/L000407/1], and the Natural
Environment Research Council [grant numbers 
NE/K008951/1 and NE/K006789/1], and by a HiPEAC collaboration
grant. The authors would like to thank Dr. Andrew T.T. McRae,
Dr. Lawrence Mitchell, and Dr. Francis Russell for their invaluable
suggestions and their contribution to the Firedrake project.

Author's addresses: Fabio Luporini $\&$ Paul H. J. Kelly, Department of Computing,
Imperial College London; David A. Ham, Department of Mathematics, Imperial College London; 
\end{bottomstuff}

\maketitle

\section{Introduction}

The need for rapid implementation of high performance, robust, and portable finite element methods has led to approaches based on automated code generation. This has been proven successful in the context of the FEniCS \cite{Fenics} and Firedrake \cite{firedrake-paper} projects. In these frameworks, the weak variational form of a problem is expressed in a high level mathematical syntax by means of the domain-specific language UFL \cite{UFL}. This mathematical specification is used by a domain-specific compiler, known as a form compiler, to generate low-level C or C++ code for the integration over a single element of the computational mesh of the variational problem's left and right hand side operators. The code for assembly operators must be carefully optimized: as the complexity of a variational form increases, in terms of number of derivatives, pre-multiplying functions, or polynomial order of the chosen function spaces, the operation count increases, with the result that assembly often accounts for a significant fraction of the overall runtime. 

As demonstrated by the substantial body of research on the topic, automating the generation of such high performance implementations poses several challenges. This is a result of the complexity inherent in the mathematical expressions involved in the numerical integration, which varies from problem to problem, and the particular structure of the loop nests enclosing the integrals. General-purpose compilers, such as those by \emph{GNU} and \emph{Intel}, fail to exploit the structure inherent in the expressions, thus producing sub-optimal code (i.e., code which performs more floating-point operations, or ``flops'', than necessary; we show this in Section~\ref{sec:perf-results}). Research compilers, for instance those based on polyhedral analysis of loop nests, such as PLUTO \cite{PLUTO}, focus on parallelization and optimization for cache locality, treating issues orthogonal to the question of minimising flops. The lack of suitable third-party tools has led to the development of a number of domain-specific code transformation (or synthesizer) systems. \citeN{quadrature1} show how automated code generation can be leveraged to introduce optimizations that a user should not be expected to write ``by hand''. \citeN{FFC-TC} and~\citeN{Francis} employ mathematical reformulations of finite element integration with the aim of minimizing the operation count. In~\citeN{Luporini}, the effects and the interplay of generalized code motion and a set of low level optimizations are analysed. It is also worth mentioning two new new form compilers, UFLACS \cite{Uflacs} and TSFC \cite{TSFC}, which particularly target the compilation time challenges of the more complex variational forms. The performance evaluation in Section~\ref{sec:perf-results} includes most of these systems.

However, in spite of such a considerable research effort, there is still no answer to one fundamental question: can we automatically generate an implementation of a form which is optimal in the number of flops executed? In this paper, we formulate an approach that solves this problem for a particular class of forms and provides very good approximations in all other cases. In particular, we will define ``local optimality'', which relates operation count with inner loops. In summary, our contributions are as follows:

\begin{itemize}
\item We formalize the class of finite element integration loop nests and we build the space of legal transformations impacting their operation count.
\item We provide an algorithm to select points in the transformation space. The algorithm uses a cost model to: (i) understand whether a transformation reduces or increases the operation count; (ii) choose between different (non-composable) transformations.
\item We demonstrate that our approach systematically leads to a local optimum. We also explain under what conditions of the input problem global optimality is achieved.
\item We integrate our approach with a compiler, COFFEE\footnote{COFFEE stands for COmpiler For Fast Expression Evaluation. The compiler is open-source and available at \url{https://github.com/coneoproject/COFFEE}}, which is in use in the Firedrake framework.
\item We experimentally evaluate using a broader suite of forms, discretizations, and code generation systems than has been used in prior research. This is essential to demonstrate that our optimality model holds in practice.
\end{itemize}

In addition, in order to place COFFEE on the same level as other code generation systems from the viewpoint of low level optimization (which is essential for a fair performance comparison):

\begin{itemize}
\item We introduce a transformation based on symbolic execution that allows irrelevant floating point operations to be skipped (for example those involving zero-valued quantities).
\end{itemize}

After reviewing basic concepts in finite element integration, in Section~\ref{sec:lnopt} we introduce a set of definitions mapping mathematical properties to the level of loop nests. This step is an essential precursor to the definition of the two algorithms -- sharing elimination (Section~\ref{sec:sharing-elimination}) and pre-evaluation (Section~\ref{sec:pre-evaluation}) -- through which we construct the space of legal transformations. The main transformation algorithm in Section~\ref{sec:optimal-synthesis} delivers the local optimality claim by using a cost model to coordinate the application of sharing elimination and pre-evaluation. We elaborate on the correctness of the methodology in Section~\ref{sec:proof}. The numerical experiments are showed in Section~\ref{sec:perf-results}. We conclude discussing the limitations of the algorithms presented and future work. 

\section{Preliminaries}
\label{sec:background}
We review finite element integration using the same notation and examples adopted in~\citeN{quadrature1} and~\citeN{Francis}. 

Consider the weak formulation of a linear variational problem:
\begin{equation}
\begin{split}
\text{Find}\ u \in U\ \text{such that} \\
a(u, v) = L(v), \forall v \in V
\end{split}
\end{equation}
where $a$ and $L$ are, respectively, a bilinear and a linear form. The set of \textit{trial} functions $U$ and the set of \textit{test} functions $V$ are suitable discrete function spaces. For simplicity, we assume $U = V$. Let $\lbrace \phi_i \rbrace$ be the set of basis functions spanning $U$. The unknown solution $u$ can be approximated as a linear combination of the basis functions $\lbrace \phi_i \rbrace$. From the solution of the following linear system it is possible to determine a set of coefficients to express $u$:
\begin{equation}
Au = b
\end{equation}
in which $A$ and $b$ discretize $a$ and $L$ respectively:
\begin{equation}
\centering
\begin{split}
A_{ij} = a(\phi_i(x), \phi_j(x)) \\
b_i = L(\phi_i(x))
\end{split}
\end{equation}
The matrix $A$ and the vector $b$ are assembled and subsequently used to solve the linear system through (typically) an iterative method.

We focus on the assembly phase, which is often characterized as a two-step procedure: \textit{local} and \textit{global} assembly. Local assembly is the subject of this article. It consists of computing the contributions of a single element in the discretized domain to the equation's approximated solution. During global assembly, these local contributions are coupled by suitably inserting them into $A$ and $b$. 

We illustrate local assembly in a concrete example, the evaluation of the local element matrix for a Laplacian operator. Consider the weighted Poisson equation:
\begin{equation}
- \nabla \cdot (w \nabla u) = 0
\end{equation}
in which $u$ is unknown, while $w$ is prescribed. The bilinear form associated with the weak variational form of the equation is:
\begin{equation}
a(v, u) = \int_\Omega w \nabla v \cdot \nabla u\ \mathrm{d}x
\end{equation}
The domain $\Omega$ of the equation is partitioned into a set of cells (elements) $T$ such that $\bigcup T = \Omega$ and $\bigcap T = \emptyset$. By defining $\lbrace \phi_i^K \rbrace$ as the set of basis functions with support on the element $K$ (i.e. those which do not vanish on this element), we can express the local element matrix as
\begin{equation}
\label{stiffness}
A_{ij}^K = \int_K w \nabla \phi_i^K \cdot \nabla \phi_j^K\ \mathrm{d}x
\end{equation}
The local element vector $L$ can be determined in an analogous way. 

\subsection{Monomials}
\label{sec:monomials}
It has been shown (for example in~\citeN{Kirby:TC}) that local element tensors can be expressed as a sum of integrals over $K$, each integral being the product of derivatives of functions from sets of discrete spaces and, possibly, functions of some spatially varying coefficients. An integral of this form is called \textit{monomial}.

\subsection{Quadrature mode}
Quadrature schemes are typically used to numerically evaluate $A_{ij}^K$. For convenience, a reference element $K_0$ and an affine mapping $F_K : K_0 \rightarrow K$ to any element $K \in T$ are introduced. This implies that a change of variables from reference coordinates $X$ to real coordinates $x = F_K (X)$ is necessary any time a new element is evaluated. The basis functions $\{\phi_i^K\}$ are then replaced with local basis functions $\{\Phi_i\}$ such that $\Phi_i(X) = \phi_i^k(F_K(X)) = \phi_i^k(x) $. The numerical integration of \eqref{stiffness} over an element $K$ can then be expressed as follows:
\begin{equation}
\label{eq:quadrature}
A_{ij}^K = \sum_{q=1}^N \sum_{\alpha_3=1}^n \Phi_{\alpha_3}(X^q)w_{\alpha_3} \sum_{\alpha_1=1}^d \sum_{\alpha_2=1}^d \sum_{\beta=1}^d \frac{\partial X_{\alpha_1}}{\partial x_{\beta}} \frac{\partial \Phi_i(X^q)}{\partial X_{\alpha_1}} \frac{\partial X_{\alpha_2}}{\partial x_{\beta}} \frac{\partial \Phi_j(X^q)}{\partial X_{\alpha_2}}  \det (F_K')W^q
\end{equation}
where $N$ is the number of integration points, $W^q$ the quadrature weight at the integration point $X^q$, $d$ the dimension of $\Omega$, $n$ the number of degrees of freedom associated to the local basis functions, and $\det (F_K')$ the determinant of the Jacobian of the aforementioned change of coordinates.

\subsection{Tensor contraction mode}
\label{sec:tc}
By exploiting the linearity, associativity and distributivity of the
relevant mathematical operators, we can rewrite \eqref{eq:quadrature} as
\begin{equation}
\label{eq:tensor}
A_{ij}^K = \sum_{\alpha_1=1}^d \sum_{\alpha_2=1}^d \sum_{\alpha_3=1}^n \left(\det (F_K') w_{\alpha_3} \sum_{\beta=1}^d \frac{X_{\alpha_1}}{\partial x_{\beta}} \frac{\partial X_{\alpha_2}}{\partial x_{\beta}} \left(\sum_{q=1}^{N} \Phi_{\alpha_3} \frac{\partial \Phi_{i_1}}{\partial X_{\alpha_1}} \frac{\partial \Phi_{i_2}}{\partial X_{\alpha_2}}W^q\right)\right).
\end{equation}
A generalization of this transformation was introduced in~\cite{Kirby:TC}. Since it only involves reference element terms, the quadrature sum can be pre-evaluated and reused for each element. The evaluation of the local tensor can then be abstracted as
\begin{equation}
A_{ij}^K = \sum_{\alpha} A_{i_1 i_2 \alpha}^0 G_{K}^\alpha
\end{equation}
in which the pre-evaluated \emph{reference tensor}, $A_{i_1 i_2 \alpha}$, and the cell-dependent \emph{geometry tensor}, $G_{K}^\alpha$, are exposed. 

\subsection{Qualitative comparison}
\label{sec:qualitative}
Depending on form and discretization, the relative performance of the two modes, in terms of the operation count, can vary quite dramatically. The presence of derivatives or coefficient functions in the input form increases the rank of the geometry tensor, making the traditional quadrature mode preferable for sufficiently complex forms. On the other hand, speed-ups from adopting tensor mode can be significant in a wide class of forms in which the geometry tensor remains sufficiently small. The discretization, particularly the polynomial order of trial, test, and coefficient functions, also plays a key role in the resulting operation count. 

These two modes are implemented in the FEniCS Form Compiler \cite{FFC-TC}. In this compiler, a heuristic is used to choose the most suitable mode for a given form. It consists of analysing each monomial in the form, counting the number of derivatives and coefficient functions, and checking if this number is greater than a constant found empirically \cite{Fenics}. We will return to the efficacy of this approach in section \ref{sec:perf-results}. One of the objectives of this paper is to produce a system that goes beyond the dichotomy between quadrature and tensor modes. We will reason in terms of loop nests, code motion, and code pre-evaluation, searching the entire implementation space for an optimal synthesis.  

\section{Transformation Space}
\label{sec:optimal-impl}
In this section, we characterize global and local optimality for finite element integration as well as the space of legal transformations that needs be explored to achieve them. The method by which exploration is performed is discussed in Section~\ref{sec:optimal-synthesis}. 

\subsection{Loop nests, expressions and optimality}
\label{sec:lnopt}
In order to make the article self-contained, we start with reviewing basic compiler terminology.

\begin{Def}[Perfect and imperfect loop nests]
A perfect loop nest is a loop whose body either 1) comprises only a sequence
of non-loop statements or 2) is itself a perfect loop nest. If this
condition does not hold, a loop nest is said to be imperfect. 
\end{Def}

\begin{Def}[Independent basic block]
An independent basic block is a sequence of statements such that no data
dependencies exist between statements in the block.
\end{Def}

We focus on perfect nests whose innermost loop body is an independent basic
block. A straightforward property of this class is that hoisting invariant
expressions from the innermost to any of the outer loops or the preheader
(i.e., the block that precedes the entry point of the nest) is always safe,
as long as any dependencies on loop indices are honored. We will make use of this property. The results of this section could also be generalized to larger classes of loop nests, in which basic block independence does not hold, although this would require refinements beyond the scope of this paper. 

By mapping mathematical properties to the loop nest level, we introduce the
concepts of a \textit{linear loop} and, more generally, a (perfect) multilinear loop nest.

\begin{Def}[Linear loop]
\label{def:linear-loop}
A loop $L$ defining the iteration space $I$ through the iteration variable $i$, or simply $L_i$, is linear if in its body
\begin{enumerate}
\item $i$ appears only as an array index, and
\item whenever an array $a$ is indexed by $i$ ($a[i]$), all expressions in which this appears are affine in $a[i]$.
\end{enumerate}
\end{Def}

\begin{Def}[Multilinear loop nest]
\label{def:multi-linear-loop}
A multilinear loop nest of arity $n$ is a perfect nest composed of $n$ loops, in which all of the expressions appearing in the body of the innermost loop are affine in each loop $L_i$ separately.
\end{Def}

We will show that multilinear loop nests, which arise naturally when translating bilinear or linear forms into code, are important because they have a structure that we can take advantage of to reach a local optimum.

We define two other classes of loops. 

\begin{Def}[Reduction loop]
\label{def:i-loop}
A loop $L_i$ is said to be a reduction loop if in its body
\begin{enumerate}
\item $i$ appears only as an array index, and
\item for each augmented assignment statement $S$ (e.g., an increment), arrays indexed by $i$ appear only on the right hand side of $S$.
\end{enumerate}
\end{Def}

\begin{Def}[Order-free loop]
\label{def:e-loop}
A loop $L_i$ is said to be an order-free loop if its iterations can be executed in any arbitrary order. 
\end{Def}

\begin{figure}\begin{CenteredBox}
\begin{lstlisting}[basicstyle=\footnotesize\ttfamily]
for (e = 0; e < E; e++)
  ...
  for (i = 0; i < I; i++)
    ...
    for (j = 0; j < J; j++)
      for (k = 0; k < K; k++)
        $a_{ejk}$ += $\displaystyle \sum_{w=1}^m \alpha_{eij}^{w} \beta_{eik}^{w} \sigma_{ei}^{w}$
\end{lstlisting}
\end{CenteredBox}\caption{The loop nest implementing a generic bilinear form.}\label{code:loopnest}\end{figure}

Consider Equation~\ref{eq:quadrature} and the (abstract) loop nest implementing it illustrated in Figure~\ref{code:loopnest}. The imperfect nest $\Lambda=[L_e, L_i, L_j, L_k]$ comprises an order-free loop $L_e$ (over elements in the mesh), a reduction loop $L_i$ (performing numerical integration), and a multilinear loop nest $[L_j, L_k]$ (over test and trial functions). In the body of $L_k$, one or more statements evaluate the local tensor for the element $e$. Expressions (the right hand side of a statement) result from the translation of a form in high level matrix notation into code. In particular, $m$ is the number of monomials (a form is a sum of monomials), $\alpha_{eij}$ ($\beta_{eik}$) represents the product of a coefficient function (e.g., the inverse Jacobian matrix for the change of coordinates) with test  or trial functions, and $\sigma_{ei}$ is a function of coefficients and geometry. We do not pose any restrictions on function spaces (e.g., scalar- or vector-valued), coefficient expressions (linear or non-linear), differential and vector operators, so $\sigma_{ei}$ can be arbitrarily complex. We say that such an expression is in \textit{normal form}, because the algebraic structure of a variational form is intact: products have not yet been expanded, distinct monomials can still be identified, and so on. This brings us to formalize the class of loop nests that we aim to optimize.

\begin{Def}[Finite element integration loop nest]
\label{def:fem-loopnest}
A finite element integration loop nest is a loop nest in which the following appear, in order: an imperfect order-free loop, an imperfect (perfect only in some special cases), linear or non-linear reduction loop, and a multilinear loop nest whose body is an independent basic block in which expressions are in normal form.
\end{Def}

We then characterize optimality for a finite element integration loop nest as follows.

\begin{Def}[Optimality of a loop nest]
\label{def:mln-optimality}
Let $\Lambda$ be a generic loop nest, and let $\Gamma$ be a transformation function $\Gamma : \Lambda \rightarrow \Lambda'$ such that $\Lambda'$ is semantically equivalent to $\Lambda$ (possibly, $\Lambda' = \Lambda$). We say that $\Lambda' = \Gamma (\Lambda)$ is an optimal synthesis of $\Lambda$ if the total number of operations (additions, products) performed to evaluate the result is minimal.
\end{Def}

The concept of local optimality, which relies on the particular class of \textit{flop-decreasing} transformations, is also introduced.

\begin{Def}[Flop-decreasing transformation]
A transformation which reduces the operation count is called flop-decreasing.
\end{Def}

\begin{Def}[Local optimality of a loop nest]
\label{def:mln-quasi-optimality}
Given $\Lambda$, $\Lambda'$ and $\Gamma$ as in Definition~\ref{def:mln-optimality}, we say that $\Lambda' = \Gamma (\Lambda)$ is a locally optimal synthesis of $\Lambda$ if:
\begin{itemize}
\item the number of operations (additions, products) in the innermost loops performed to evaluate the result is minimal, and
\item $\Gamma$ is expressed as composition of flop-decreasing transformations.
\end{itemize}
\end{Def}

The restriction to flop-decreasing transformations aims to exclude those apparent optimizations that, to achieve flop-optimal innermost loops, would rearrange the computation at the level of the outer loops causing, in fact, a global increase in operation count. 

We also observe that Definitions~\ref{def:mln-optimality} and~\ref{def:mln-quasi-optimality} do not take into account memory requirements. If the execution of loop nest were memory-bound -- the ratio of operations to bytes transferred from memory to the CPU being too low -- then optimizing the number of flops would be fruitless. Henceforth we assume we operate in a CPU-bound regime, evaluating arithmetic-intensive expressions. In the context of finite elements, this is often true for more complex multilinear forms and/or higher order elements. 

Achieving optimality in polynomial time is not generally feasible, since the $\sigma_{ei}$ sub-expressions can be arbitrarily unstructured. However, multilinearity results in a certain degree of regularity in $\alpha_{eij}$ and $\beta_{eik}$. In the following sections, we will elaborate on these observations and formulate an approach that achieves: (i) at least a local optimum in all cases; (ii) global optimality whenever the monomials are ``sufficiently structured''. To this purpose, we will construct:
\begin{itemize}
\item the space of legal transformations impacting the operation count (Sections~\ref{sec:sharing-elimination} -- \ref{sec:mem-const})
\item an algorithm to select points in the transformation space (Section~\ref{sec:optimal-synthesis})
\end{itemize}

\subsection{Sharing elimination}
\label{sec:sharing-elimination}
We start with introducing the fundamental notion of sharing.

\begin{Def}[Sharing]
A statement within a loop nest $\Lambda$ presents sharing if at least one of the following conditions hold:
\begin{description}
\item[Spatial sharing] There are at least two symbolically identical sub-expressions
\item[Temporal sharing] There is at least one non-trivial sub-expression (e.g., an addition or a product) that is redundantly executed because it is independent of $\lbrace L_{i_1}, L_{i_1}, ...L_{i_n} \rbrace \subset \Lambda$.
\end{description}
\end{Def}

To illustrate the definition, we show in Figure~\ref{code:multi_loopnest} how sharing evolves as factorization and code motion are applied to a trivial multilinear loop nest. In the original loop nest (Figure~\ref{code:multi_loopnest_a}), spatial sharing is induced by the symbol $b_j$. Factorization eliminates spatial sharing and creates temporal sharing (Figure~\ref{code:multi_loopnest_b}). Finally, generalized code motion \cite{Luporini}, which hoists sub-expressions that are redundantly executed by at least one loop in the nest\footnote{Traditional loop-invariant code motion, which is commonly applied by general-purpose compilers, only checks invariance with respect to the innermost loop.}, leads to optimality (Figure~\ref{code:multi_loopnest_c}). 

\begin{figure}[h]\begin{CenteredBox}
{\subfigcapskip = 13pt \subfigure[With spatial sharing]{\label{code:multi_loopnest_a}
\begin{lstlisting}[basicstyle=\footnotesize\ttfamily]
for (j = 0; j < J; j++)
  for (i = 0; i < I; i++)
    $a_{ji}$ += $b_j c_i$ + $b_j d_i$
\end{lstlisting}}}
~~~~~
{\subfigcapskip = 13pt \subfigure[With temporal sharing]{\label{code:multi_loopnest_b}
\begin{lstlisting}[basicstyle=\footnotesize\ttfamily]
for (j = 0; j < J; j++)
  for (i = 0; i < I; i++)
    $a_{ji}$ += $b_j (c_i$ + $d_i$)
\end{lstlisting}}}
~~~~~
{\subfigcapskip = 4pt \subfigure[Optimal form]{\label{code:multi_loopnest_c}
\begin{lstlisting}[basicstyle=\footnotesize\ttfamily]
for (i = 0; i < I; i++)
  $t_i$ = $c_i + d_i$
for (j = 0; j < J; j++)
  for (i = 0; i < I; i++)
    $a_{ji}$ += $b_j t_i$
\end{lstlisting}}}
\end{CenteredBox}\caption{Reducing a simple multilinear loop nest to optimal form.}\label{code:multi_loopnest}\end{figure}

In this section, we study \textit{sharing elimination}, a transformation that aims to reduce the operation count by removing sharing through the application of expansion, factorization, and generalized code motion. If the objective were reaching optimality and the expressions lacked structure, a transformation of this sort would require solving a large combinatorial problem -- for instance to evaluate the impact of all possible factorizations. Our sharing elimination strategy, instead, exploits the structure inherent in finite element integration expressions to guarantee, after coordination with other transformations (an aspect which we discuss in the following sections), local optimality.  Global optimality is achieved if stronger preconditions hold. Setting local optimality, rather than optimality, as primary goal is essential to produce simple and computationally efficient algorithms -- two necessary conditions for integration with a compiler.

\subsubsection{Identification and exploitation of structure}
\label{sec:se-rln}
Finite element expressions can be seen as composition of operations between tensors. Often, the optimal implementation strategy for these operations is to be determined out of two alternatives. For instance, consider $J^{-T} \nabla v \cdot J^{-T} \nabla v$, with $J^{-T}$ being the transposed inverse Jacobian matrix for the change of (two-dimensional) coordinates, and $v$ a generic two-dimensional vector. The tensor operation will reduce to the scalar expression $(a v^0_i + b v^1_i) (a v^0_i + b v^1_i) + ...$, in which $v^0_i$ and $v^1_i$ represent components of $v$ that depend on $L_i$. To minimize the operation count for expressions of this kind, we have two options:
\begin{Strategy}
\label{strategy:i}
Eliminating temporal sharing through generalized code motion.
\end{Strategy}
\begin{Strategy}
\label{strategy:ii}
Eliminating spatial sharing first -- through product expansion and factorization -- and temporal sharing afterwards, again through generalized code motion.
\end{Strategy}
In the current example, we observe that, depending on the size of $L_i$, applying Strategy~\ref{strategy:ii} could reduce the operation count since the expression would be recast as $v^0_i v^0_i a a + v^0_i v^1_i (ab + ab) + v^1_i v^1_i c c + ...$ and some hoistable sub-expressions would be exposed. On the other hand, Strategy~\ref{strategy:i} would have no effect as $v$ only depends on a single loop, $L_i$. In general, the choice between the two strategies depends on multiple factors: the loop sizes, the increase in operation count due to expansion (in Strategy~\ref{strategy:ii}), and the gain due to code motion. A second application of Strategy~\ref{strategy:ii} was provided in Figure~\ref{code:multi_loopnest}. These examples motivate the introduction of a particular class of expressions, for which the two strategies assume notable importance.
\begin{Def}[Structured expression]
\label{def:struct-expr}
We say that an expression is ``structured along a loop nest $\Lambda$'' if and only if,  for every symbol $s_{\Lambda}$ depending on at least one loop in $\Lambda$, the spatial sharing of $s_{\Lambda}$ may be eliminated by factorizing all occurrences of $s_{\Lambda}$ in the expression.
\end{Def}
\begin{Prop}
\label{prop:multi-struct}
An expression along a multilinear loop nest is structured.
\end{Prop}
\begin{proof}
This follows directly from Definition~\ref{def:linear-loop} and Definition~\ref{def:multi-linear-loop}, which essentially restrict the number of occurrences of a symbol $s_{\Lambda}$ in a summand to at most 1.
\end{proof}
If $\Lambda$ were an arbitrary loop nest, a given symbol $s_{\Lambda}$ could appear everywhere (e.g., $n$ times in a summand and $m$ times in another summand with $n \neq m$, as argument of a higher level function, in the denominator of a division), thus posing the challenge of finding the factorization that maximizes temporal sharing. If $\Lambda$ is instead a finite element integration loop nest, thanks to Proposition~\ref{prop:multi-struct} the space of flop-decreasing transformations is constructed by ``composition'' of Strategy~\ref{strategy:i} and Strategy~\ref{strategy:ii}, as illustrated in Algorithm~\ref{algo:sharing-elimination}.

Finally, we observe that the $\sigma_{ei}$ sub-expressions can sometimes be considered ``weakly structured''. This happens when a relaxed version of Definition~\ref{def:struct-expr} applies, in which the factorization of $s_{\Lambda}$ only ``minimizes'' (rather than ``eliminates'') spatial sharing (for instance, in the complex hyperelastic model analyzed in Section~\ref{sec:perf-results}). Weak structure will be exploited by Algorithm~\ref{algo:sharing-elimination} in the attempt to achieve optimality.

\subsubsection{Algorithm}
\label{sec:se-algo}
Algorithm~\ref{algo:sharing-elimination} describes sharing elimination assuming as input a tree representation of the loop nest. It makes use of the following notation and terminology:
\begin{itemize}
\item \textit{multilinear operand}: any $\alpha_{eij}$ or $\beta_{eik}$ in the input expression.
\item \textit{multilinear symbol}: a symbol appearing within a multilinear operand depending on $L_j$ or $L_k$ (e.g., test functions, first order derivatives of test functions, etc.).
\end{itemize}
Examples will be provided in Section~\ref{sec:se-examples}.

\noindent\rule[0.01ex]{\linewidth}{0.7pt}

\begin{Algo}[Sharing elimination]
\label{algo:sharing-elimination}
\normalfont 
The input of the algorithm is a tree representation a finite element integration loop nest.
\begin{enumerate}
\item Perform a depth-first visit of the loop tree to collect and partition multilinear operands into disjoint sets, $\mathbb{P} = \lbrace P_1, ..., P_p\rbrace$. $\mathbb{P}$ is such that all multilinear operands in each $P \in \mathbb{P}$ share the same set of multilinear symbols $S_P$, whereas there is no sharing across different partitions. For all multilinear operands in $P \in \mathbb{P}$ such that $|P| \leq |S_P|$, apply Strategy~\ref{strategy:i}. \\
\textit{Note: as a consequence of Proposition~\ref{prop:multi-struct}, $|P|$ and $|S_P|$ represent the number of products in the innermost loop induced by $P$ if Strategy~\ref{strategy:i} or Strategy~\ref{strategy:ii} were applied}

\item For each sub-expression $\operatorname{expr}$ depending on exactly one linear loop, collect the multilinear symbols and the temporaries produced at step (1). Partition them into disjoint sets, $\mathbb{T} = \lbrace T_1, ..., T_t\rbrace$, such that $T_i$ includes all instances of a given symbol in $\operatorname{expr}$. Apply Strategy~\ref{strategy:ii} factorizing the symbols in each $T_i$, provided that this leads to a reduction in operation count; otherwise, apply Strategy~\ref{strategy:i} \\
\textit{Note: the last check ensures the flop-decreasing nature of the transformation. In the cases in which expansion outweighs code motion, Strategy~\ref{strategy:i} is preferred.}\\
\textit{Note: the expansion cost is a function of the products wrapping a symbol (how many of them and their arity), so it can be determined through tree visits.}

\item Build the \textit{sharing graph} $G = (S, E)$. Each $s \in S$ represents a multilinear symbol or a temporary produced by the previous steps. An edge $(s_i$, $s_j)$ indicates that a product $s_i s_j$ would appear if the sub-expressions including $s_i$ and $s_j$ were expanded.\\
\textit{Note: the following steps will only impact bilinear forms, since otherwise $E = \emptyset$.}

\item Partition $S$ into disjoint sets, $\mathbb{S} = \lbrace S_1, ..., S_n\rbrace$, such that $S_i$ includes all instances of a given symbol $s$ in the expression. Transform $G$ by merging $\lbrace s_1, ..., s_m \rbrace \subset S_i$ into a unique vertex $s$ (taking the union of the edges), provided that factorizing $[s_1, ..., s_m]$ would not cause an increase in operation count.

\item Map $G$ to an Integer Linear Programming (ILP) model for determining how to optimally apply Strategy~\ref{strategy:ii}. The solution is the set of symbols that will be factorized by Strategy~\ref{strategy:ii}. Let $|S| = n$; the ILP model then is as follows:
\begin{gather*}
x_i \text{: a vertex in } S \text{ (1 if a symbol should be factorized, 0 otherwise)}\\
y_{ij} \text{: an edge in } E \text{ (1 if } s_i \text{ is factorized in the product } s_i s_j \text{ , 0 otherwise)}\\
n_i \text{: the number of edges incident to } x_i \\
\begin{align*}
\min \sum_{i=1}^{n} x_i,\ s.t. ~~~~~~~~~&\sum_{j|(i,j) \in E} y_{ij} \leq n_i x_i,\ i = 1, ..., n \\
& y_{ij} + y_{ji} = 1,\ (i, j) \in E
\end{align*}
\phantom{\hspace{6cm}}
\end{gather*}

\item Perform a depth-first visit of the loop tree and, for each yet unhandled or hoisted expression, apply the most profitable between Strategy~\ref{strategy:i} and Strategy~\ref{strategy:ii}. \\
\textit{Note: this pass speculatively assumes that expressions are (weakly) structured along the reduction loop. If the assumption does not hold, the operation count will generally be sub-optimal because only a subset of factorizations and code motion opportunities may eventually be considered.}
\end{enumerate}
\end{Algo}

\noindent\rule[1.0ex]{\linewidth}{0.7pt}

Although the primary goal of Algorithm~\ref{algo:sharing-elimination} is operation count minimization within the multilinear loop nest, the enforcement of flop-decreasing transformations (steps (2) and (4)) and the re-scheduling of sub-expressions within outer loops (last step) also attempt to optimize the loop nest globally. We will further elaborate this aspect in Section~\ref{sec:proof}.

\subsubsection{Examples}
\label{sec:se-examples}
Consider again Figure~\ref{code:multi_loopnest_a}. We have $\mathbb{P} = \lbrace P_0, P_1, P_2\rbrace$, with $P_0 = \lbrace b_j\rbrace$, $P_1 = \lbrace c_i\rbrace$, and $P_2 = \lbrace d_i\rbrace$. For all $P_i$, we have $|P_i| = 1 = |S_{P_i}|$, although applying Strategy~\ref{strategy:i} in step (1) has no effect. The sharing graph is $G = (\lbrace b_j, c_i, d_i \rbrace, \lbrace (b_j, c_i), (b_j, d_i) \rbrace$, and $T = \lbrace c_i, d_i \rbrace$. The ILP formulation leads to the code in Figure~\ref{code:multi_loopnest_c}.

In Figure~\ref{code:poisson}, Algorithm~\ref{algo:sharing-elimination} is executed in a very simple realistic scenario, which originates from the bilinear form of a Poisson equation in two dimensions. We observe that $\mathbb{P} = \lbrace P_0, P_1\rbrace$, with $P_0 = \lbrace (z_0 a_{ik} + z_2 b_{ik}), (z_1 a_{ik} + z_3 b_{ik})\rbrace$ and $P_1 = \lbrace (z_0 a_{ij} + z_2 b_{ij}), (z_1 a_{ij} + z_3 b_{ij})\rbrace$. In addition, $|P_i| = |S_{P_i}| = 2$, so Strategy~\ref{strategy:i} is applied to both partitions (step (1)). We then have (step (3)) $G = (\lbrace t_0, t_1, t_2, t_3 \rbrace, \lbrace (t_0, t_2), (t_1, t_3) \rbrace)$. Since there are no more factorization opportunities, the ILP formulation becomes irrelevant.

\begin{figure}[h]\begin{CenteredBox}
{\subfigcapskip = 7pt \subfigure[Normal form]{\label{code:poisson_a}
\begin{lstlisting}[basicstyle=\footnotesize\ttfamily]
for (e = 0; e < E; e++)
  $z_0$ = ...
  $z_1$ = ...
  ...
  for (i = 0; i < I; i++)
    for (j = 0; j < J; j++)
      for (k = 0; k < K; k++)
        $a_{ejk}$ += $(((z_0 a_{ik} + z_2 b_{ik})*$
                 $(z_0 c_{ij} + z_2 d_{ij}))+$
                $((z_1 a_{ik} + z_3 b_{ik})*$
                 $(z_1 c_{ij} + z_3 d_{ij})))*$
                $W_i*\operatorname{det}$
\end{lstlisting}}}
~~~~~~~~~~
{\subfigcapskip = 7pt \subfigure[After sharing elimination]{\label{code:poisson_b}
\begin{lstlisting}[basicstyle=\footnotesize\ttfamily]
for (e = 0; e < E; e++)
  ...
  for (i = 0; i < I; i++)
    for (k = 0; k < K; k++)
      $t_{0_k}$ = $(z_0 a_{ik} + z_2 b_{ik})$
      $t_{1_k}$ = $(z_1 a_{ik} + z_3 b_{ik})$
    for (j = 0; j < J; j++)
      $t_{2_j}$ = $(z_0 c_{ij} + z_2 d_{ij})*W_i*\operatorname{det}$
      $t_{3_j}$ = $(z_1 c_{ij} + z_3 d_{ij})*W_i*\operatorname{det}$
    for (j = 0; j < J; j++)
      for (k = 0; k < K; k++)
        $a_{ejk}$ += $t_{0_k} * t_{2_j}$ + $t_{1_k} * t_{3_j}$
\end{lstlisting}}}
\end{CenteredBox}\caption{Applying sharing elimination to the bilinear form arising from a Poisson equation in 2D. The operation counts are $E(f(z_0, z_1, ...) + IJK \cdot 18)$ (left) and $E(f(z_0, z_1, ...) + I(J \cdot 6 + K \cdot 9 + JK \cdot 4)$ (right), with $f(z_0, z_1, ...)$ representing the operation count for evaluating $z_0, z_1, ...$, and common sub-expressions being counted once. The synthesis in Figure~\ref{code:poisson_b} is globally optimal apart from the pathological case  $I, J, K = 1$.}\label{code:poisson}\end{figure}

For reasons of space, further examples, including the hyperelastic model evaluated in Section~\ref{sec:perf-results} and other non-trivial ILP instances, are made available online\footnote{Sharing elimination examples: \url{https://gist.github.com/FabioLuporini/14e79457d6b15823c1cd}}.

\subsection{Pre-evaluation of reductions}
\label{sec:pre-evaluation}
Sharing elimination uses three operators: expansion, factorization, and code motion. In this section, we discuss the role and legality of a fourth operator: reduction pre-evaluation. We will see that what makes this operator special is the fact that there exists a single point in the transformation space of a monomial (i.e., a specific factorization of test, trial, and coefficient functions) ensuring its correctness.

We start with an example. Consider again the loop nest and the expression in Figure~\ref{code:loopnest}. We pose the following question: are we able to identify sub-expressions for which the reduction induced by $L_i$ can be pre-evaluated, thus obtaining a decrease in operation count proportional to the size of $L_i$, $I$? The transformation we look for is exemplified in Figure~\ref{code:loopnest_rednored} with a simple loop nest. The reader may verify that a similar transformation is applicable to the example in Figure~\ref{code:poisson_a}.

\begin{figure}[h]\begin{CenteredBox}
{\subfigcapskip = 19pt \subfigure[With reduction]{\label{code:loopnest_red}
\begin{lstlisting}[basicstyle=\footnotesize\ttfamily]
for (e = 0; e < E; e++)
  for (i = 0; i < I; i++)
    for (k = 0; k < K; k++)
      $a_{ek}$ += $d_e b_{ik} c_i$ + $d_e b_{ik} d_i$
\end{lstlisting}}}
~~~~~~~~~~
{\subfigcapskip = 5pt \subfigure[After pre-evaluation]{\label{code:loopnest_nored}
\begin{lstlisting}[basicstyle=\footnotesize\ttfamily]
for (i = 0; i < I; i++)
  for (k = 0; k < K; k++)
    $t_k$ += $b_{ik}$($c_i$ + $d_i$)

for (e = 0; e < E; e++)
  for (k = 0; k < K; k++)
    $a_{ek}$ = $d_e t_k$
\end{lstlisting}}}
\end{CenteredBox}\caption{Exposing (through factorization) and pre-evaluating a reduction.}\label{code:loopnest_rednored}\end{figure}

Pre-evaluation can be seen as the generalization of tensor contraction (Section~\ref{sec:tc}) to a wider class of sub-expressions. We know that multilinear forms can be seen as sums of monomials, each monomial being an integral over the equation domain of products (of derivatives) of functions from discrete spaces. A monomial can always be reduced to the product between a ``reference'' and a ``geometry'' tensor. In our model, a reference tensor is simply represented by one or more sub-expressions independent of $L_e$, exposed after particular transformations of the expression tree. This leads to the following algorithm. 

\noindent\rule[0.01ex]{\linewidth}{0.7pt}

\begin{Algo}[Pre-evaluation]
\label{algo:pre-evaluation}
\normalfont
Consider a finite element integration loop nest $\Lambda = [L_e, L_i, L_j, L_k]$. We dissect the normal form input expression into distinct sub-expressions, each of them representing a monomial. Each sub-expression is then factorized so as to split constants from $[L_i, L_j, L_k]$-dependent terms. This transformation is feasible\footnote{For reasons of space, we omit the detailed sequence of steps (e.g., expansion, factorization), which is however available at \url{https://github.com/coneoproject/COFFEE/blob/master/coffee/optimizer.py} in \citeN{fabio_luporini_2016_49279}.}, as a consequence of the results in~\citeN{Kirby:TC}. These $[L_i, L_j, L_k]$-dependent terms are hoisted outside of $\Lambda$ and stored into temporaries. As part of this process, the reduction induced by $L_i$ is computed by means of symbolic execution. Finally, $L_i$ is removed from $\Lambda$. 
\end{Algo}

\noindent\rule[1.0ex]{\linewidth}{0.7pt}

The pre-evaluation of a monomial introduces some critical issues:
\begin{enumerate}
\item Depending on the complexity of a monomial, a certain number, $t$, of temporary variables is required if pre-evaluation is performed. Such temporary variables are actually $n$-dimensional arrays of size $S$, with $n$ and $S$ being, respectively, the arity and the extent (iteration space size) of the multilinear loop nest (e.g., $n=2$ and $S = J K$ in the case of bilinear forms). For certain values of ${<}t, n, S{>}$, pre-evaluation may dramatically increase the working set, which may be counter-productive for actual execution time.
\item The transformations exposing $[L_i, L_j, L_k]$-dependent terms increase the arithmetic complexity of the expression (e.g., expansion tends to increase the operation count). This could outweigh the gain due to pre-evaluation.
\item A strategy for coordinating sharing elimination and pre-evaluation is needed. We observe that sharing elimination inhibits pre-evaluation, whereas pre-evaluation could expose further sharing elimination opportunities.
\end{enumerate}

We expand on point (1) in the next section, while we address points (2) and (3) in Section~\ref{sec:optimal-synthesis}. 

\subsection{Memory constraints}
\label{sec:mem-const}
We have just observed that the code motion induced by monomial pre-evaluation may dramatically increase the working set size. Even more aggressive code motion strategies are theoretically conceivable. Imagine $\Lambda$ is enclosed in a time stepping loop. One could think of exposing (through some transformations) and hoisting time-invariant sub-expressions for minimizing redundant computation at each time step. The working set size would then increase by a factor $E$, and since $E \gg I, J, K$, the gain in operation count would probably be outweighed, from a runtime viewpoint, by a much larger memory pressure.

Since, for certain forms and discretizations, hoisting may cause the working set to exceed the size of some level of local memory (e.g. the last level of private cache on a conventional CPU, the shared memory on a GPU), we introduce the following \textit{memory constraints}.

\begin{Const}
\label{const:Le}
The size of a temporary due to code motion must not be proportional to the size of $L_e$.
\end{Const}

\begin{Const}
\label{const:TH}
The total amount of memory occupied by the temporaries due to code motion must not exceed a certain threshold, \texttt{$T_H$}.
\end{Const}

Constraint~\ref{const:Le} is a policy decision that the compiler should not silently consume memory on global data objects. It has the effect of shrinking the transformation space. Constraint~\ref{const:TH} has both theoretical and practical implications, which will be carefully analyzed in the next sections.

\section{Selection and composition of transformations}
\label{sec:optimal-synthesis}
In this section, we build a transformation algorithm that, given a memory bound, systematically reaches a local optimum for finite element integration loop nests. 

\subsection{Transformation algorithm}
We address the two following issues: 
\begin{enumerate}
\item \textit{Coordination of pre-evaluation and sharing elimination.} Recall from Section~\ref{sec:pre-evaluation} that pre-evaluation could either increase or decrease the operation count in comparison with that achieved by sharing elimination.
\item \textit{Optimizing over composite operations.} Consider a form comprising two monomials $m_1$ and $m_2$. Assume that pre-evaluation is profitable for $m_1$ but not for $m_2$, and that $m_1$ and $m_2$ share at least one term (for example some basis functions). If pre-evaluation were applied to $m_1$, sharing between $m_1$ and $m_2$ would be lost. We then need a mechanism to understand which transformation -- pre-evaluation or sharing elimination -- results in the highest operation count reduction when considering the whole set of monomials (i.e., the expression as a whole).
\end{enumerate}

Let $\theta : M \rightarrow \mathbb{Z}$ be a cost function that, given a monomial $m \in M$, returns the gain/loss achieved by pre-evaluation over sharing elimination. In particular, we define $\theta(m) = \mathrm{\theta}^{pre}(m) - \mathrm{\theta}^{se}(m)$, where $\mathrm{\theta}^{se}$ and $\mathrm{\theta}^{pre}$ represent the operation counts resulting from applying sharing elimination and pre-evaluation, respectively. Thus pre-evaluation is profitable for $m$ if and only if $\theta(m) < 0$. We return to the issue of deriving $\mathrm{\theta}^{se}$ and $\mathrm{\theta}^{pre}$ in Section~\ref{sec:op_count}. Having defined $\theta$, we can now describe the transformation algorithm (Algorithm~\ref{algo:gamma}).

\noindent\rule[0.01ex]{\linewidth}{0.7pt}

\begin{Algo}[Transformation algorithm]
\label{algo:gamma}
\normalfont
The algorithm has three main phases: initialization (step 1); determination of the monomials preserving the memory constraints that should be pre-evaluated (steps 2-4); application of pre-evaluation and sharing elimination (step 5).
\begin{enumerate}
\item Perform a depth-first visit of the expression tree and determine the set of monomials $M$. Let $S$ be the subset of monomials $m$ such that $\theta(m) > 0$. The set of monomials that will \textit{potentially} be pre-evaluated is $P = M \setminus S$. \\ \textit{Note: there are two fundamental reasons for not pre-evaluating $m_1 \in P$ straight away: 1) the potential presence of spatial sharing between $m_1$ and $m_2 \in S$, which impacts the search for the global optimum; 2) the risk of breaking Constraint~\ref{const:TH}.}
\item Build the set $B$ of all possible bipartitions of $P$. Let $D$ be the dictionary that will store the operation counts of different alternatives.
\item Discard $b = (b_S, b_P) \in B$ if the memory required after applying pre-evaluation to the monomials in $b_P$ exceeds $T_H$ (see Constraint~\ref{const:TH}); otherwise, add $D[b] = \mathrm{\theta}^{se}(S \cup b_S) + \mathrm{\theta}^{pre}(b_P)$. \\ \textit{Note: $\mathbb{B}$ is in practice very small, since even complex forms usually have only a few monomials. This pass can then be accomplished rapidly as long as the cost of calculating $\mathrm{\theta}^{se}$ and $\mathrm{\theta}^{pre}$ is negligible. We elaborate on this aspect in Section~\ref{sec:op_count}.}
\item Take $\argmin_b D[b]$.
\item Apply pre-evaluation to all monomials in $b_P$. Apply sharing elimination to all resulting expressions. \\ \textit{Note: because of the reuse of basis functions, pre-evaluation may produce some identical tables, which will be mapped to the same temporary variable. Sharing elimination is therefore transparently applied to all expressions, including those resulting from pre-evaluation.}
\end{enumerate}
\end{Algo}

\noindent\rule[1.0ex]{\linewidth}{0.7pt}

The output of the transformation algorithm is provided in Figure~\ref{code:loopnest-opt}, assuming as input the loop nest in Figure~\ref{code:loopnest}. 

\begin{figure}[h]\begin{CenteredBox}
\begin{lstlisting}[basicstyle=\footnotesize\ttfamily]
// Pre-evaluated tables
...
for (e = 0; e < E; e++)
  // Temporaries due to sharing elimination
  // (Sharing was a by-product of pre-evaluation)
  ...
  // Loop nest for pre-evaluated monomials
  for (j = 0; j < J; j++)
    for (k = 0; k < K; k++)
      $a_{ejk}$ += F$'$(...) + F$''$(...) + ...

  // Loop nest for monomials for which run-time
  // integration was determined to be faster
  for (i = 0; i < I; i++)
    // Temporaries due to sharing elimination
    ...
    for (j = 0; j < J; j++)
      for (k = 0; k < K; k++)
        $a_{ejk}$ += H(...)
\end{lstlisting}
\end{CenteredBox}\caption{The loop nest produced by the algorithm for an input as in Figure~\ref{code:loopnest}.}\label{code:loopnest-opt}\end{figure}

\subsection{The cost function $\theta$}
\label{sec:op_count}
We tie up the remaining loose end: the construction of the cost function $\theta$.

We recall that $\theta(m) = \mathrm{\theta}^{se}(m) - \mathrm{\theta}^{pre}(m)$, with $\mathrm{\theta}^{se}$ and $\mathrm{\theta}^{pre}$ representing the operation counts after applying sharing elimination and pre-evaluation. Since $\theta$ is deployed in a working compiler, simplicity and efficiency are essential characteristics. In the following, we explain how to derive these two values.

The most trivial way of evaluating $\mathrm{\theta}^{se}$ and $\mathrm{\theta}^{pre}$ would consist of applying the actual transformations and simply count the number of operations. This would be tolerable for $\mathrm{\theta}^{se}$, as Algorithm~\ref{algo:sharing-elimination} tends to have negligible cost. However, the overhead would be unacceptable if we applied pre-evaluation -- in particular, symbolic execution -- to all bipartitions analyzed by Algorithm~\ref{algo:gamma}. We therefore seek an analytic way of determining $\mathrm{\theta}^{pre}$.

The first step consists of estimating the \textit{increase factor}, $\iota$. This number captures the increase in arithmetic complexity due to the transformations exposing pre-evaluation opportunities. For context, consider the example in Figure~\ref{code:increase_factor}. One can think of this as the (simplified) loop nest originating from the integration of the action of a mass matrix. The sub-expression \texttt{$f_0$*$B_{i0}$+$f_1$*$B_{i1}$+$f_2$*$B_{i2}$} represents the coefficient $f$ over (tabulated) basis functions (array $B$). In order to apply pre-evaluation, the expression needs be transformed to separate $f$ from all $[L_i, L_j, L_k]$-dependent quantities (see Algorithm~\ref{algo:pre-evaluation}). By product expansion, we observe an increase in the number of $[L_j, L_k]$-dependent terms of a factor $\iota = 3$.

\begin{figure}[h]\begin{CenteredBox}
\begin{lstlisting}[basicstyle=\footnotesize\ttfamily]
for (i = 0; i < I; i++)
  for (j = 0; j < J; j++)
    for (k = 0; k < K; k++)
      $a_{jk}$ += $b_{ij}$*$b_{ik}$*($f_0$*$B_{i0}$+$f_1$*$B_{i1}$+$f_2$*$B_{i2}$)
\end{lstlisting}
\end{CenteredBox}\caption{Simplified loop nest for a pre-multiplied mass matrix.}\label{code:increase_factor}\end{figure}

In general, however, determining $\iota$ is not so straightforward since redundant tabulations may result from common sub-expressions. Consider the previous example. One may add one coefficient in the same function space as $f$, repeat the expansion, and observe that multiple sub-expressions (e.g., $b_{10}*b_{01}*...$ and $b_{01}*b_{10}*...$) will reduce to identical tables. To evaluate $\iota$, we then use combinatorics. We calculate the $k$-combinations with repetitions of $n$ elements, where: (i) $k$ is the number of (derivatives of) coefficients appearing in a product; (ii) $n$ is the number of unique basis functions involved in the expansion. In the original example, we had $n=3$ (for $b_{i0}$, $b_{i1}$, and $b_{i2}$) and $k=1$, which confirms $\iota=3$. In the modified example, there are two coefficients, so $k=2$, which means $\iota=6$.

If $\iota \geq I$ (the extent of the reduction loop), we already know that pre-evaluation will not be profitable. Intuitively, this means that we are introducing more operations than we are saving from pre-evaluating $L_i$. If $\iota < I$, we still need to find the number of terms $\rho$ such that $\mathrm{\theta}^{pre} = \rho \cdot \iota$. The mass matrix monomial in Figure~\ref{code:increase_factor} is characterized by the dot product of test and trial functions, so trivially $\rho = 1$. In the example in Figure~\ref{code:poisson}, instead, we have $\rho = 3$ after a suitable factorization of basis functions. In general, therefore, $\rho$ depends on both form and discretization. To determine this parameter, we look at the re-factorized expression (as established by Algorithm~\ref{algo:pre-evaluation}), and simply count the terms amenable to pre-evaluation.

\section{Formalization}
\label{sec:proof}
We demonstrate that the orchestration of sharing elimination and pre-evaluation performed by the transformation algorithm guarantees local optimality (Definition~\ref{def:mln-quasi-optimality}). The proof re-uses concepts and explanations provided throughout the paper, as well as the terminology introduced in Section~\ref{sec:se-algo}.

\begin{Prop}
\label{prop:optimal-approach}
Consider a multilinear form comprising a set of monomials $M$, and let $\Lambda$ be the corresponding finite element integration loop nest. Let $\Gamma$ be the transformation algorithm. Let $X$ be the set of monomials that, according to $\Gamma$, need to be pre-evaluated, and let $Y = M \setminus X$. Assume that the pre-evaluation of different monomials does not result in identical tables. Then, $\Lambda' = \Gamma(\Lambda)$ is a local optimum in the sense of Definition~\ref{def:mln-quasi-optimality} and satisfies Constraint~\ref{const:TH}.
\end{Prop}
\begin{proof}
We first observe that the cost function $\theta$ predicts the \textit{exact} gain/loss in monomial pre-evaluation, so $X$ and $Y$ can actually be constructed.

Let $c_\Lambda$ denote the operation count for $\Lambda$ and let $\Lambda_I \subset \Lambda$ be the subset of innermost loops (all $L_k$ loops in Figure~\ref{code:loopnest-opt}). We need to show that there is no other synthesis $\Lambda''_{I}$ satisfying Constraint~\ref{const:TH} such that $c_{\Lambda''_I} < c_{\Lambda'_I}$. This holds if and only if
\begin{enumerate}
\item \textit{The coordination of pre-evaluation with sharing elimination is optimal}. This boils down to prove that
\begin{enumerate}
\item \textit{pre-evaluating any $m \in Y$ would result in $c_{\Lambda''_I} > c_{\Lambda'_I}$}
\item \textit{not pre-evaluating any $m \in X$ would result in $c_{\Lambda''_I} > c_{\Lambda'_I}$}
\end{enumerate}
\item \textit{Sharing elimination leads to a (at least) local optimum.}\\
\end{enumerate}

We discuss these points separately

\begin{enumerate}
\item 
\begin{enumerate}
\item Let $T_m$ represent the set of tables resulting from applying pre-evaluation to a monomial $m$. Consider two monomials $m_1, m_2 \in Y$ and the respective sets of pre-evaluated tables, $T_{m_{1}}$ and $T_{m_{2}}$. If $T_{m_{1}} \cap\ T_{m_{2}} \neq \emptyset$, at least one table is assignable to the same temporary. $\Gamma$, therefore, may not be optimal, since $\theta$ only distinguishes monomials in ``isolation''. We neglect this scenario (see assumptions) because of its purely pathological nature and its -- with high probability -- negligible impact on the operation count.
\item Let $m_1 \in X$ and $m_2 \in Y$ be two monomials sharing some generic multilinear symbols. If $m_1$ were carelessly pre-evaluated, there may be a potential gain in sharing elimination that is lost, potentially leading to a non-optimum. This situation is prevented by construction, because $\Gamma$ exhaustively searches all possible bipartitions on order to determine an optimum which satisfies Constraint~\ref{const:TH}\footnote{Note that the problem can be seen as an instance of the well-known Knapsack problem}. Recall that since the number of monomials is in practice very small, this pass can rapidly be accomplished.
\end{enumerate}
\item Consider Algorithm~\ref{algo:sharing-elimination}. Proposition~\ref{prop:multi-struct} ensures that there are only two ways of scheduling the multilinear operands in $P \in \mathbb{P}$: through generalized code motion (Strategy~\ref{strategy:i}) or factorization of multilinear symbols (via Strategy~\ref{strategy:ii}). If applied, these two strategies would lead, respectively, to performing $|P|$ and $|S_P|$ multiplications at every loop iteration. Since Strategy~\ref{strategy:i} is applied if and only if $|P| < |S_P|$ and does not change the structure of the expression (it requires neither expansion nor factorization), step (1) cannot prune the optimum from the search space. 

After structuring the sharing graph $G$ in such a way that only flop-decreasing transformations are possible, the ILP model is instantiated. At this point, proving optimality reduces to establishing the correctness of the model, which is relatively straightforward because of its simplicity. The model aims to minimize the operation count by selecting the most promising factorizations. The second set of constraints is to select all edges (i.e., all multiplications), exactly once. The first set of inequalities allows multiplications to be scheduled: once a vertex $s$ is selected (i.e., once a symbol is decided to be factorized), all multiplications involving $s$ can be grouped. 
\end{enumerate}
\end{proof}

Throughout the paper we have reiterated the claim that Algorithm~\ref{algo:gamma} achieves a globally optimal flop count if stronger preconditions on the input variational form are satisfied. We state here these preconditions, in increasing order of complexity.
\begin{enumerate}
\item There is a single monomial and only a specific coefficient (e.g., the coordinates field). This is by far the simplest scenario, which requires no particular transformation at the level of the outer loops, so optimality naturally follows.
\item There is a single monomial, but multiple coefficients are present. Optimality is achieved if and only if all sub-expressions depending on coefficients are structured (see Section~\ref{sec:se-rln}). This avoids ambiguity in factorization, which in turn guarantees that the output of step (7) in Algorithm~\ref{algo:sharing-elimination} is optimal.
\item There are multiple monomials, but either at most one coefficient (e.g., the coordinates field) or multiple coefficients not inducing sharing across different monomials are present. This reduces, respectively, to cases (1) and (2) above.
\item There are multiple monomials, and coefficients are shared across monomials. Optimality is reached if and only if the coefficient-dependent sub-expressions produced by Algorithm~\ref{algo:sharing-elimination} -- that is, the by-product of factorizing test/trial functions from distinct monomials -- preserve structure. 
\end{enumerate}

\section{Code Generation}
\label{sec:codegen}
Sharing elimination and pre-evaluation, as well as the transformation algorithm, have been implemented in COFFEE, the compiler for finite element integration routines adopted in Firedrake. In this section, we briefly discuss the aspects of the compiler that are relevant for this article.

\subsection{Expressing transformations through the COFFEE language}
COFFEE implements sharing elimination and pre-evaluation by composing building block transformation operators, which we refer to as \emph{rewrite operators}. This has several advantages. The first is extensibility. New transformations, such as sum factorization in spectral methods, could be expressed by composing the existing operators, or with small effort building on what is already available. Second, generality: COFFEE can be seen as a lightweight, low level computer algebra system, not necessarily tied to finite element integration. Third, robustness: the same operators are exploited, and therefore tested, by different optimization pipelines. The rewrite operators, whose (Python) implementation is based on manipulation of abstract syntax trees (ASTs), comprise the COFFEE language. A non-exhaustive list of such operators includes expansion, factorization, re-association, generalized code motion.

\subsection{Independence from form compilers}
COFFEE aims to be independent of the high level form compiler. It provides an interface to build generic ASTs and only expects expressions to be in normal form (or sufficiently close to it). For example, Firedrake has transitioned from a version of the FEniCS Form Compiler \cite{FFC-TC} modified to produce ASTs rather than strings, to a newly written compiler\footnote{TSFC, the two-stage form compiler \url{https://github.com/firedrakeproject/tsfc}}, while continuing to emply COFFEE. Thus, COFFEE decouples the mathematical manipulation of a form from code optimization; or, in other words, relieves form compiler developers of the task of fine scale loop optimization of generated code.

\subsection{Handling block-sparse tables}
\label{sec:zeros}
For several reasons, basis function tables may be block-sparse (e.g., containing zero-valued columns). For example, the FEniCS Form Compiler implements vector-valued functions by adding blocks of zero-valued columns to the corresponding tabulations; this extremely simplifies code generation (particularly, the construction of loop nests), but also affects the performance of the generated code due to the execution of ``useless'' flops (e.g., operations like \texttt{a + 0}). In~\citeN{quadrature1}, a technique to avoid iteration over zero-valued columns based on the use of indirection arrays (e.g. \texttt{A[B[i]]}, in which \texttt{A} is a tabulated basis function and \texttt{B} a map from loop iterations to non-zero columns in A) was proposed. This technique, however, produces non-contiguous memory loads and stores, which nullify the potential benefits of vectorization. COFFEE, instead, handles block-sparse basis function tables by restructuring loops in such a manner that low level optimization (especially vectorization) is only marginally affected. This is based on symbolic execution of the code, which enables a series of checks on array indices and loop bounds which determine the zero-valued blocks which can be skipped without affecting data alignment.

\section{Performance Evaluation}
\label{sec:perf-results}

\subsection{Experimental setup}

Experiments were run on a single core of an Intel I7-2600 (Sandy Bridge) CPU, running at 3.4GHz, 32KB L1 cache (private), 256KB L2 cache (private) and 8MB L3 cache (shared). The Intel Turbo Boost and Intel Speed Step technologies were disabled. The Intel \texttt{icc 15.2} compiler was used. The compilation flags used were \texttt{-O3, -xHost}. The compilation flag \texttt{xHost} tells the Intel compiler to generate efficient code for the underlying platform.

The Zenodo system was used to archive all packages used to perform the experiments: Firedrake \cite{lawrence_mitchell_2016_49284}, PETSc \cite{barry_smith_2016_49285}, petsc4py \cite{firedrake_2016_49283}, FIAT \cite{marie_e_rognes_2016_49280}, UFL \cite{martin_sandve_alnaes_2016_49282}, FFC \cite{anders_logg_2016_49276}, PyOP2 \cite{florian_rathgeber_2016_49281} and COFFEE \cite{fabio_luporini_2016_49279}. The experiments can be reproduced using a publicly available benchmark suite~\cite{florian_rathgeber_2016_49290}.

We analyze the execution time of four real-world bilinear forms of increasing complexity, which comprise the differential operators that are most common in finite element methods. In particular, we study the mass matrix (``\texttt{Mass}'') and the bilinear forms arising in a Helmholtz equation (``\texttt{Helmholtz}''), in an elastic model (``\texttt{Elasticity}''), and in a hyperelastic model (``\texttt{Hyperelasticity}''). The complete specification of these forms is made publicly available\footnote{\url{https://github.com/firedrakeproject/firedrake-bench/blob/experiments/forms/firedrake_forms.py}}. 

We evaluate the speed-ups achieved by a wide variety of transformation systems over the ``original'' code produced by the FEniCS Form Compiler (i.e., no optimizations applied). We analyze the following transformation systems:
\begin{description}
\item[quad] Optimized quadrature mode. Work presented in~\citeN{quadrature1}, implemented in  in the FEniCS Form Compiler. 
\item[tens] Tensor contraction mode. Work presented in~\citeN{FFC-TC}, implemented in the FEniCS Form Compiler.
\item[auto] Automatic choice between \texttt{tens} and \texttt{quad} driven by heuristic (detailed in~\citeN{Fenics} and summarized in Section~\ref{sec:qualitative}). Implemented in the FEniCS Form Compiler.
\item[ufls] UFLACS, a novel back-end for the FEniCS Form Compiler whose primary goals are improved code generation and execution times.
\item[cfO1] Generalized loop-invariant code motion. Work presented in~\citeN{Luporini}, implemented in COFFEE.
\item[cfO2] Optimal loop nest synthesis with handling of block-sparse tables. Work presented in this article, implemented in COFFEE.
\end{description}

The values that we report are the average of three runs with ``warm cache''; that is, with all kernels retrieved directly from the Firedrake's cache, so code generation and compilation times are not counted. The timing includes however the cost of both local assembly and matrix insertion, with the latter minimized through the choice of a mesh (details below) small enough to fit the L3 cache of the CPU. 

For a fair comparison, small patches were written to the make \texttt{quad}, \texttt{tens}, and \texttt{ufls} compatible with Firedrake. By executing all simulations in Firedrake, we guarantee that both matrix insertion and mesh iteration have a fixed cost, independent of the transformation system employed. The patches adjust the data storage layout to what Firedrake expects (e.g., by generating an array of pointers instead of a pointer to pointers, by replacing flattened arrays with bi-dimensional ones). 

For Constraint~\ref{const:TH}, discussed in Section~\ref{sec:mem-const}, we set $T_H = \operatorname{size}(\mathrm{L2})$; that is, the size of the processor L2 cache (the last level of private cache). When the threshold had an impact on the transformation process, the experiments were repeated with $T_H = \operatorname{size}(\mathrm{L3})$. The results are documented later, individually for each problem.

Following the methodology adopted in~\citeN{quadrature1}, we vary the following parameters:
\begin{itemize}
\item the polynomial degree of test, trial, and coefficient (or ``pre-multiplying'') functions, $q \in \lbrace1, 2, 3, 4\rbrace$
\item the number of coefficient functions $\mathrm{nf} \in \lbrace0, 1, 2, 3\rbrace$
\end{itemize}
While constants of our study are
\begin{itemize}
\item the space of test, trial, and coefficient functions: Lagrange
\item the mesh: tetrahedral with a total of 4374 elements
\item exact numerical quadrature (we employ the same scheme used in~\citeN{quadrature1}, based on the Gauss-Legendre-Jacobi rule)
\end{itemize}

\subsection{Performance results}
\label{sec:perf-results-forms}

\begin{figure}
 \makebox[\textwidth][c]{\includegraphics[scale=0.77]{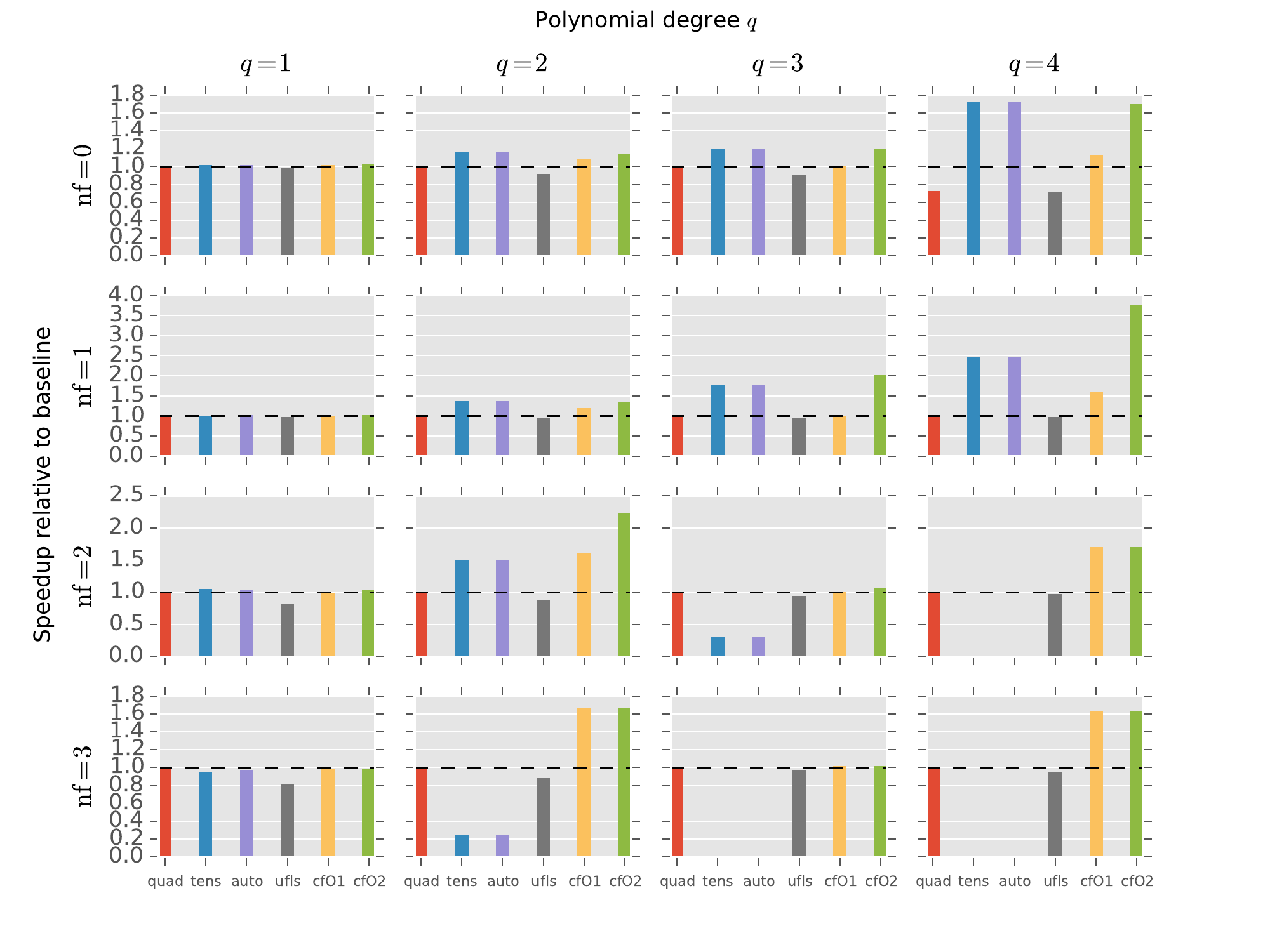}}
\caption{Performance evaluation for the \textit{mass} matrix. The bars represent speed-up over the original (unoptimized) code produced by the FEniCS Form Compiler.}\label{fig:mass}
\end{figure}

\begin{figure}
 \makebox[\textwidth][c]{\includegraphics[scale=0.77]{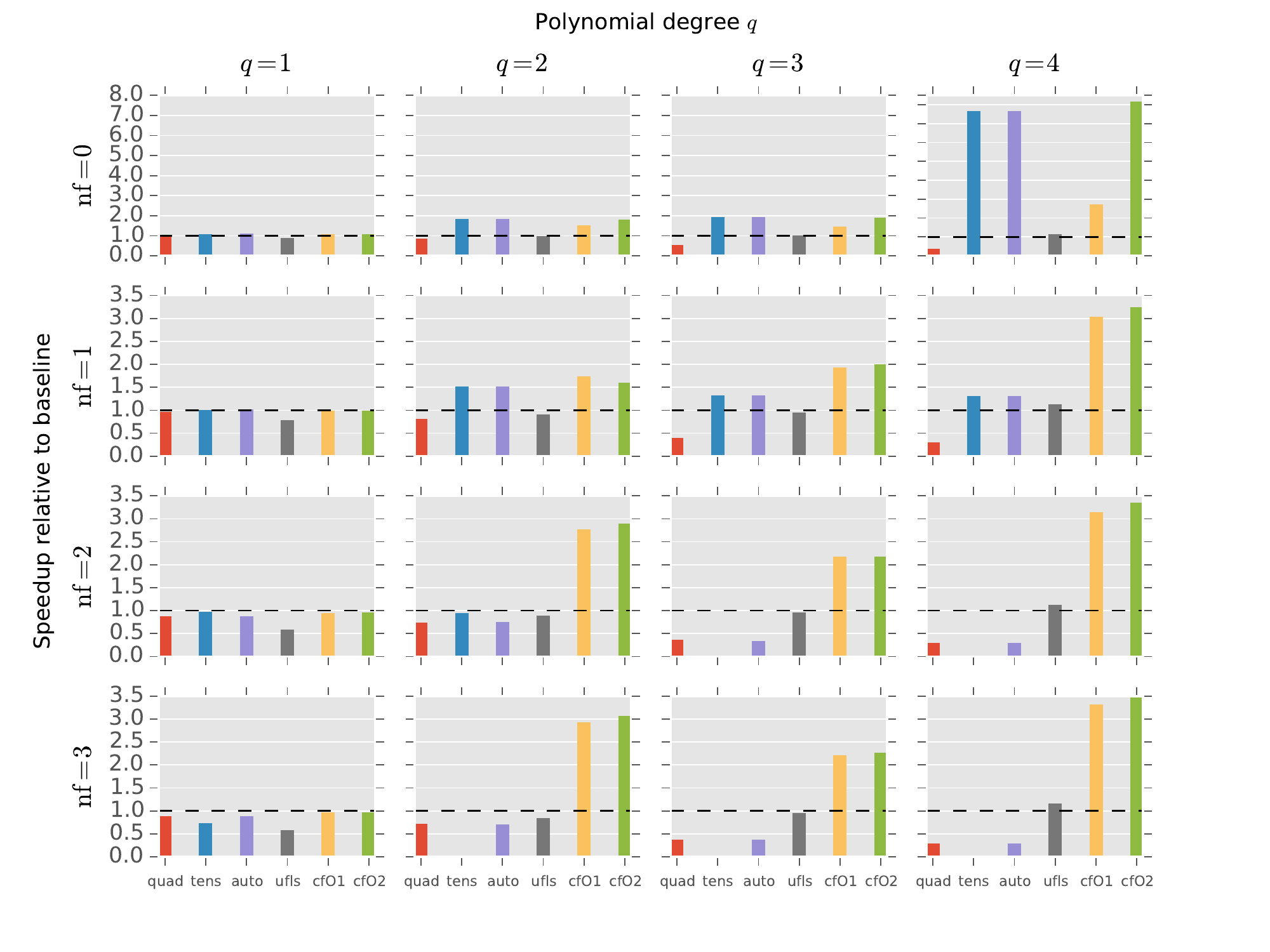}}
\caption{Performance evaluation for the bilinear form of a \textit{Helmholtz} equation. The bars represent speed-up over the original (unoptimized) code produced by the FEniCS Form Compiler.}\label{fig:helmholtz}
\end{figure}

\begin{figure}
\includegraphics[scale=0.77]{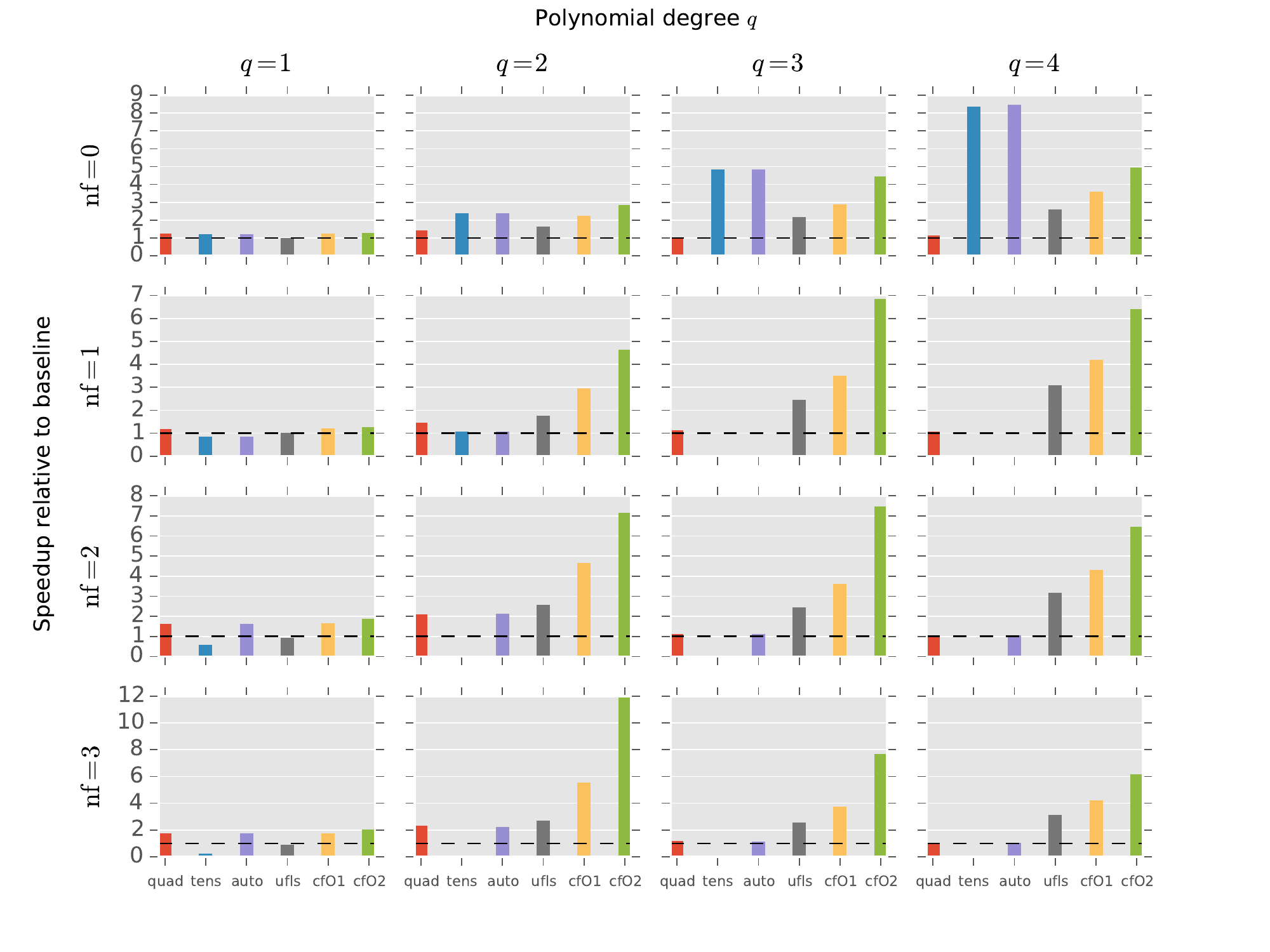}
\caption{Performance evaluation for the bilinear form arising in an \textit{elastic} model. The bars represent speed-up over the original (unoptimized) code produced by the FEniCS Form Compiler.}\label{fig:elasticity}
\end{figure}

\begin{figure}
 \makebox[\textwidth][c]{\includegraphics[scale=0.77]{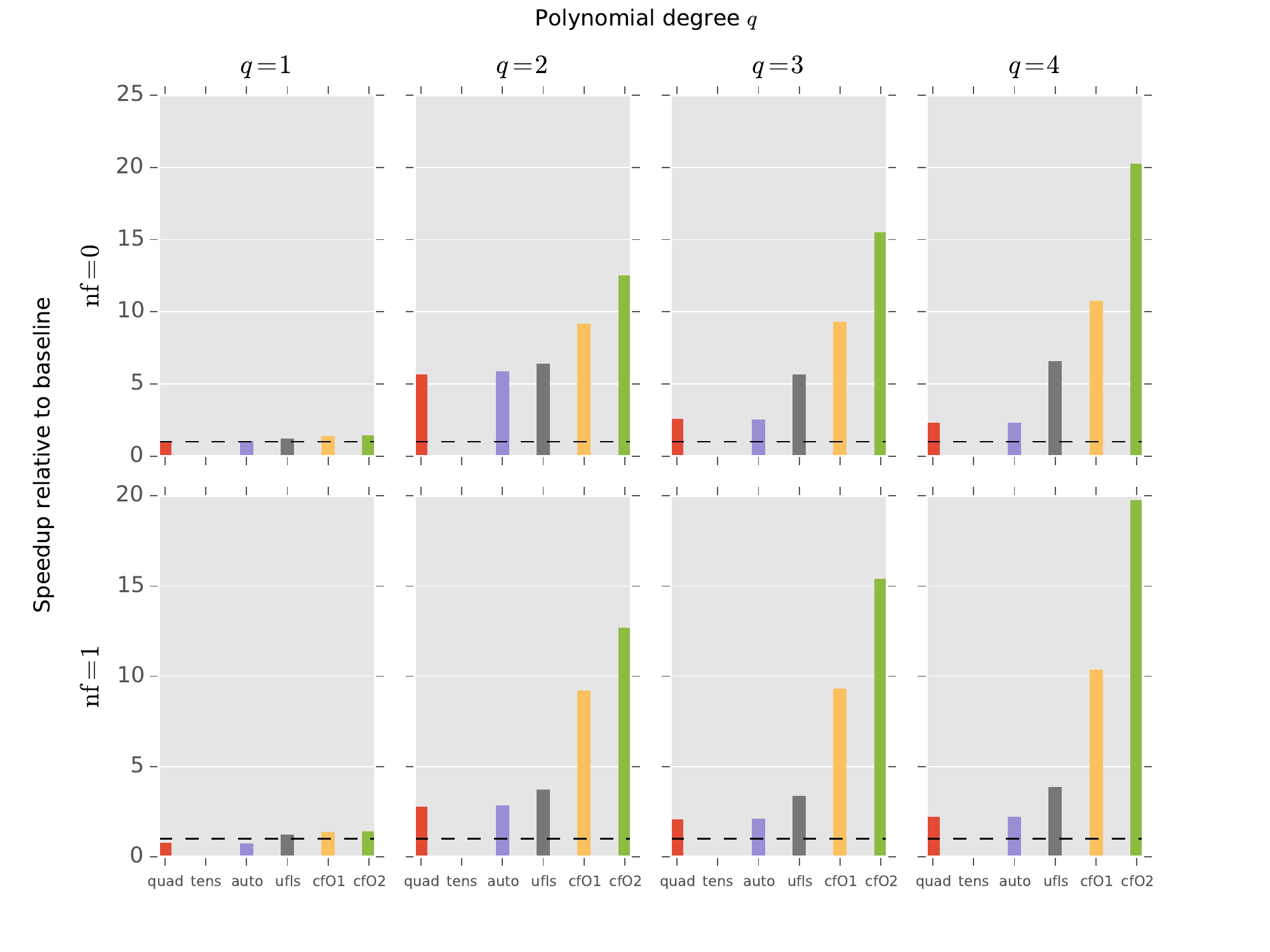}}
\caption{Performance evaluation for the bilinear form arising in a \textit{hyperelastic} model. The bars represent speed-up over the original (unoptimized) code produced by the FEniCS Form Compiler.}\label{fig:hyperelasticity}
\end{figure}

We report the results of our experiments in Figures~\ref{fig:mass},~\ref{fig:helmholtz},~\ref{fig:elasticity}, and~\ref{fig:hyperelasticity} as three-dimensional plots. The axes represent $q$, $\mathrm{nf}$, and code transformation system. We show one subplot for each problem instance ${\langle} \textrm{form}, \mathrm{nf}, q {\rangle}$, with the code transformation system varying within each subplot. The best variant for each problem instance is given by the tallest bar, which indicates the maximum speed-up over non-transformed code. We note that if a bar or a subplot are missing, then the form compiler failed to generate code because it either exceeded the system memory limit or was otherwise unable to handle the form. 

The rest of the section is organized as follows: we first provide insights into the general outcome of the experimentation; we then comment on the impact of a fundamental low-level optimization, namely autovectorization; finally, we motivate, for each form, the performance results obtained.

\paragraph{High level view}
Our transformation strategy does not always guarantee minimum execution time. In particular, about 5$\%$ of the test cases (3 out of 56, without counting marginal differences) show that \texttt{cfO2} was not optimal in terms of runtime. The most significant of such test cases is the elastic model with $[q=4,\ \mathrm{nf}=0]$. There are two reasons for this. First, low level optimization can have a significant impact on the actual performance. For example, the aggressive loop unrolling in \texttt{tens} eliminates operations on zeros and reduces the working set size by not storing entire temporaries; on the other hand, preserving the loop structure can maximize the chances of autovectorization. Second, the transformation strategy adopted when $T_H$ is exceeded plays a key role, as we will later elaborate.

\paragraph{Autovectorization}
 
We chose the mesh dimension and the function spaces such that the inner loop sizes would always be a multiple of the machine vector length. This ensured autovectorization in the majority of code variants\footnote{We verified the vectorization of inner loops by looking at both compiler reports and assembly code.}. The biggest exception is \texttt{quad}, due to the presence of indirection arrays in the generated code. In \texttt{tens}, loop nests are fully unrolled, so the standard loop vectorization is not feasible; the compiler reports suggest, however, that block vectorization \cite{SLP-vect} is often triggered. In \texttt{ufls}, \texttt{cfO1}, and \texttt{cfO2} the iteration spaces have identical structure, with loop vectorization being regularly applied.

\paragraph{Mass matrix}
We start with the simplest of the bilinear forms investigated, the mass matrix. Results are in Figure~\ref{fig:mass}. We first notice that the lack of improvements when $q=1$ is due to the fact that matrix insertion outweighs local assembly. For $q \geq 2$, \texttt{cfO2} generally shows the highest speed-ups. It is worth noting why \texttt{auto} does not always select the fastest implementation: \texttt{auto} always opts for \texttt{tens}, while for $\mathrm{nf} \geq 2$ \texttt{quad} tends to be preferable. On the other hand, \texttt{cfO2} always makes the optimal decision about whether to apply pre-evaluation or not. Surprisingly, despite the simplicity of the form, the performance of the various code generation systems can differ significantly.

\paragraph{Helmholtz}
As in the case of Mass matrix, when $q=1$ the matrix insertion phase is dominant. For $q \geq 2$, the general trend is that \texttt{cfO2} outperforms the competitors. In particular:
\begin{description}
\item[$\mathrm{nf}=0$] pre-evaluation makes \texttt{cfO2} notably faster than \texttt{cfO1}, especially for high values of $q$; \texttt{auto} correctly selects \texttt{tens}, which is comparable to \texttt{cfO2}. 
\item[$\mathrm{nf}=1$] \texttt{auto} picks \texttt{tens}; the choice is however sub-optimal when $q=3$ and $q=4$. This can indirectly be inferred from the large gap between \texttt{cfO2} and \texttt{tens/auto}: \texttt{cfO2} applies sharing elimination, but it correctly avoids pre-evaluation because of the excessive expansion cost.
\item[$\mathrm{nf}=2$ and $\mathrm{nf}=3$] \texttt{auto} reverts to \texttt{quad}, which would theoretically be the right choice (the flop count is much lower than in \texttt{tens}); however, the generated code suffers from the presence of indirection arrays, which break autovectorization and ``traditional'' code motion.
\end{description}

The slow-downs (or marginal improvements) seen in a small number of cases exhibited by \texttt{ufls} can be attributed to the presence of sharing in the generated code.

An interesting experiment we additionally performed was relaxing the memory threshold by setting $T_H = \operatorname{size}(\mathrm{L3})$. We found that this makes \texttt{cfO2} generally slower for $\mathrm{nf} \geq 2$, with a maximum slow-down of 2.16$\times$ with $\langle \mathrm{nf}=2, q=2\rangle$. This effect could be worse when running in parallel, since the L3 cache is shared and different threads would end up competing for the same resource.

\paragraph{Elasticity}
The results for the elastic model are displayed in Figure~\ref{fig:elasticity}. The main observation is that \texttt{cfO2} never triggers pre-evaluation, although in some occasions it should. To clarify this, consider the test case $\langle \mathrm{nf}=0,\ q=2 \rangle$, in which \texttt{tens/auto} show a considerable speed-up over \texttt{cfO2}. \texttt{cfO2} finds pre-evaluation profitable in terms of operation count, although it is eventually not applied to avoid exceeding $T_H$. However, running the same experiments with $T_H = \operatorname{size}(\mathrm{L3})$ resulted in a dramatic improvement, even higher than that obtained by \texttt{tens}. The reason is that, despite exceeding $T_H$ by roughly 40$\%$, the saving in operation count is so large (5$\times$ in this specific problem) that pre-evaluation would in practice be the winning choice. This suggests that our objective function should be improved to handle the cases in which there is a significant gap between potential cache misses and reduction in operation count.

We also note that:
\begin{itemize}
\item the differences between \texttt{cfO2} and \texttt{cfO1} are due to the perfect sharing elimination and the zero-valued blocks avoidance technique presented in Section~\ref{sec:zeros}.
\item when $\mathrm{nf}=1$, \texttt{auto} prefers \texttt{tens} over \texttt{quad}, which leads to sub-optimal operation counts and execution times.
\item \texttt{ufls} often results in better execution times than \texttt{quad} and \texttt{tens}. This is due to multiple factors, including avoidance of indirection arrays, preservation of loop structure, and a more effective code motion strategy.
\end{itemize}

\paragraph{Hyperelasticity}
In the experiments on the hyperelastic model, shown in Figure~\ref{fig:hyperelasticity}, \texttt{cfO2} exhibits the largest gains out of all problem instances considered in this paper. This is a positive result, since it indicates that our transformation algorithm scales well with form complexity. The fact that all code transformation systems (apart from \texttt{tens}) show quite significant speed-ups suggests two points. First, the baseline is highly inefficient. With forms as complex as in the hyperelastic model, a trivial translation of integration routines into code should always be avoided as even the best general-purpose compiler available (the Intel compiler on an Intel platform at maximum optimization level) fails to exploit the structure inherent in the expressions. Second, the strategy for removing spatial and temporal sharing has a tremendous impact. Sharing elimination as performed by \texttt{cfO2} ensures a critical reduction in operation count, which becomes particularly pronounced for higher values of $q$. 

\section{Conclusions}
\label{sec:conclusions}
We have developed a theory for the optimization of finite element integration loop nests. The article details the domain properties which are exploited by our approach (e.g., linearity) and how these translate to transformations at the level of loop nests. All of the algorithms shown in this paper have been implemented in COFFEE, a compiler publicly available fully integrated with the Firedrake framework. The correctness of the transformation algorithm was discussed. The performance results achieved suggest the effectiveness of our methodology. 

\section{Limitations and future work}
\label{sec:completeness}
We have defined sharing elimination and pre-evaluation as high level transformations on top of a specific set of rewrite operators, such as code motion and factorization, and we have used them to construct the transformation space. There are three main limitations in this process. First, we do not have a systematic strategy to optimize sub-expressions which are independent of linear loops. Although we have a mechanism to determine how much computation should be hoisted to the level of the integration (reduction) loop, it is not clear how to effectively improve the heuristics used at step (6) in Algorithm~\ref{algo:sharing-elimination}. Second, lower operation counts may be found by exploiting domain-specific properties, such as redundancies in basis functions; this aspect is completely neglected in this article. Third, with Constraint~\ref{const:Le} we have limited the applicability of code motion. This constraint was essential given the complexity of the problem tackled. 

Another issue raised by the experimentation concerns selecting a proper threshold for Constraint~\ref{const:TH}. To solve this problem would require a more sophisticated cost model, which is an interesting question deserving further research. 

We also identify two additional possible research directions: a complete classification of forms for which a global optimum is achieved; and a generalization of the methodology to other classes of loop nests, for instance those arising in spectral element methods.

\bibliographystyle{ACM-Reference-Format-Journals}

\medskip

\end{document}